\documentclass[submission,copyright,creativecommons]{eptcs}
\providecommand{\event}{QPL 2026} 

\usepackage{iftex}

\ifpdf
  \usepackage{underscore}         
  \usepackage[T1]{fontenc}        
\else
  \usepackage{breakurl}           
\fi

\usepackage{amsmath,amssymb,amsthm}
\usepackage{mathtools}
\usepackage{amssymb}
\usepackage{cmll}
\usepackage{comment}
\usepackage{graphicx}
\usepackage{physics}
\usepackage{stmaryrd}
\usepackage{tikz}
\usepackage{tikz-cd}
\usepackage{tikzit}
\usepackage{circuitikz}
\usepackage{tubes}
\usepackage{xsavebox}

\input{styles.tikzdefs}

\tikzstyle{discarding}=[fill=white, draw=black, shape=circle, style=upground]
\tikzstyle{smalldiscarding}=[fill=white, draw=black, style=upground]
\tikzstyle{backdiscard}=[fill=white, draw=black, shape=circle, style=downground]
\tikzstyle{smallbackdiscard}=[fill=white, draw=black, shape=circle, style=downground]
\tikzstyle{state}=[fill=white, draw=black, style=triang, tikzit shape=rectangle]
\tikzstyle{kstate}=[fill=white, draw=black, style=kpoint, tikzit shape=rectangle]
\tikzstyle{kstateconj}=[fill=white, draw=black, style=kpoint conjugate, tikzit shape=rectangle]
\tikzstyle{kstateBIG}=[fill=white, draw=black, style=big kpoint, tikzit shape=rectangle]
\tikzstyle{effect}=[fill=white, draw=black, style=triangdag]
\tikzstyle{keffect}=[fill=white, draw=black, style=kpoint adjoint]
\tikzstyle{keffectconj}=[fill=white, draw=black, style=kpoint transpose]
\tikzstyle{morphdag}=[style=mapdag]
\tikzstyle{morph}=[style=hadamard]
\tikzstyle{WIDEmorph}=[style=hadamard, minimum width=14mm]
\tikzstyle{morphtrans}=[style=maptrans]
\tikzstyle{morphconj}=[style=mapconj]
\tikzstyle{CPMmorph}=[style=dmap]
\tikzstyle{CPMmorphconj}=[style=dmapconj]
\tikzstyle{CPMmorphdag}=[style=dmapdag]
\tikzstyle{CPMmorphtrans}=[style=dmaptrans]
\tikzstyle{CPMstate}=[fill=white, draw=black, style=triang, doubled]
\tikzstyle{CPMstateBIG}=[fill=white, draw=black, style={triang_lesssep}, doubled]
\tikzstyle{CPMkstate}=[fill=white, draw=black, style=kpoint, tikzit shape=rectangle, doubled]
\tikzstyle{CPMkstateconj}=[fill=white, draw=black, style=kpoint conjugate, tikzit shape=rectangle, doubled]
\tikzstyle{CPMkstateBIG}=[fill=white, draw=black, style=big kpoint, tikzit shape=rectangle, doubled]
\tikzstyle{CPMkeffect}=[fill=white, draw=black, style=kpoint adjoint, doubled]
\tikzstyle{CPMkeffectconj}=[fill=white, draw=black, style=kpoint transpose, doubled]
\tikzstyle{UHfB}=[fill=white, draw=black, style=triangdag, doubled, inner sep=-2pt]
\tikzstyle{leak}=[style=tinypoint, regular polygon rotate=-90]
\tikzstyle{leakfill}=[style=tinypoint, regular polygon rotate=-90, fill=black]
\tikzstyle{Z}=[style=dot, fill=green]
\tikzstyle{X}=[style=dot, fill=red]
\tikzstyle{black_dot}=[style=dot, fill=black]
\tikzstyle{white_dot}=[style=dot, fill=white]
\tikzstyle{qblack_dot}=[style=ddot, fill=black]
\tikzstyle{qwhite_dot}=[style=ddot, fill=white]
\tikzstyle{whitephase}=[style=wphase dot, fill=white]
\tikzstyle{qredphase}=[style=phase dot, fill=red]
\tikzstyle{qgreenphase}=[style=phase dot, fill=green]
\tikzstyle{had}=[style=hadamard, doubled]
\tikzstyle{box}=[style=hadamard]
\tikzstyle{bigbox}=[style=hadamard, minimum height=4mm, minimum width=8mm]
\tikzstyle{classhad}=[style=hadamard]
\tikzstyle{antipode}=[style=anti]

\tikzstyle{dottededge}=[-, dash pattern=on 1pt off 0.7pt]
\tikzstyle{double edge}=[-, style=doubled, draw=black, tikzit draw={rgb,255: red,18; green,168; blue,191}]
\tikzstyle{new edge style 0}=[<-]
\tikzstyle{new edge style 1}=[-, draw={rgb,255: red,242; green,233; blue,206}, fill={rgb,255: red,242; green,233; blue,206}]
\tikzstyle{morphism_shade}=[-, draw=black, fill={rgb,255: red,242; green,233; blue,206}, line join=bevel]
\tikzstyle{supermap_shade}=[-, fill={rgb,255: red,216; green,215; blue,242}, draw=black, line join=bevel]
\tikzstyle{hole_shade}=[-, fill=white, draw=black,line join=bevel]
\tikzstyle{new edge style 2}=[-, draw={rgb,255: red,14; green,188; blue,83}]
\tikzstyle{new edge style 3}=[<-, draw={rgb,255: red,234; green,209; blue,255}]
\tikzstyle{new edge style 4}=[<-, draw={rgb,255: red,0; green,128; blue,128}]
\tikzstyle{new edge style 5}=[-, draw={rgb,255: red,214; green,110; blue,62}]
\tikzstyle{new edge style 6}=[-, draw={rgb,255: red,174; green,20; blue,174}]

\newcommand{\tikzfigscale}[2]{\scalebox{#1}{\input{#2.tikz}}}

\newcommand{\morph}[1]{\xrightarrow{#1}}

\newcommand{\pmorph}{\relbar\joinrel\mapstochar\joinrel\rightarrow}
\newcommand{\cat}[1]{\mathbf{#1}}

\newcommand{\opcat}[1]{#1^\textrm{op}}
\newcommand{\set}{\mathbf{Set}}

\newcommand{\StProf}{\mathbf{StProf}}
\newcommand{\StEnv}{\mathbf{StEnv}}
\newcommand{\Chu}{\mathbf{Chu}}

\newcommand{\opt}{\mathbf{Optic}}

\newcommand{\caus}{\mathbf{Caus}}

\newcommand{\bl}{\mathord{=}}

\newcommand{\seq}{\varolessthan}

\newcommand{\TODO}[1]{\marginpar{\scriptsize{==TODO==\\#1}}}

\newcommand{\C}{\cat{C}}
\newcommand{\F}{\mathcal{F}}

\DeclareFontFamily{U}{min}{}
\DeclareFontShape{U}{min}{m}{n}{<-> udmj30}{}


\newtheorem{theorem}{Theorem}[section]

\newtheorem{proposition}[theorem]{Proposition}
\theoremstyle{definition}
\newtheorem{definition}{Definition}
\theoremstyle{remark}
\newtheorem{remark}[theorem]{Remark}

\pgfsetlayers{background,strings,edgelayer,nodelayer,foreground,main}

\title{Higher-Order Quantum Objects are Strong Profunctors}

\author{Matt Wilson
\institute{Universit\'e Paris-Saclay, CNRS, ENS Paris-Saclay, Inria, Laboratoire M\'ethodes Formelles\\
91190 Gif-sur-Yvette, France}
\email{matthew.wilson@centralesupelec.fr}
\and
James Hefford
\institute{Universit\'e Paris-Saclay, CNRS, ENS Paris-Saclay, Inria, Laboratoire M\'ethodes Formelles\\
91190 Gif-sur-Yvette, France}
\email{james.hefford@inria.fr}
}

\def\titlerunning{Higher-Order Quantum Objects are Strong Profunctors}
\def\authorrunning{M. Wilson \& J. Hefford}

\begin{document}
\maketitle

\begin{abstract}
We explore the sense in which the existing constructions for higher-order maps on quantum theory based on causality constraints and compositionality constraints respectively, coincide. 
More precisely, we construct a functor $\mathcal{F}: \caus(\mathbf{C}) \rightarrow \StProf(\mathbf{C}_1)$ from higher-order causal categories to the category of strong profunctors over first-order causal processes that is lax-lax duoidal, full, faithful, and strongly closed whenever $\mathbf{C}$ is additive. 
When $\mathbf{C} = \mathbf{CP}$ this embedding is furthermore strong on the sequencer for duoidal categories, expressing the possibility to interpret one-way signalling (but not general non-signalling) constraints in terms of the coend calculus for profunctors.
We conclude that insofar as compositional constraints can be used to express causality constraints, the profunctorial approach generalises higher-order quantum theory to a construction over general symmetric monoidal categories. 
\end{abstract}

\section{Introduction}

A basic principle in the foundations of physics is to try to use \textit{compositionality} to study \textit{causality}. Firstly, in the study of causal decompositions, a core theorem on quantum information theory is the realisation that \textit{semi-causality} entails \textit{semi-localizability} both in factor \cite{eggeling, beckman_localisable} and non-factor systems \cite{mestoudjian2026picturinggeneralquantumsubsystems}, a result which can be formalised with the following diagrams representing an equivalence between a causality constraint and a compositionality constraint:  \[ \left(\tikzfigscale{1}{figs/caus_comp_3} \ = \   \tikzfigscale{1}{figs/caus_comp_4}\right) \ \iff \ \left(\tikzfigscale{1}{figs/caus_comp_1} \ = \ \tikzfigscale{1}{figs/caus_comp_2} \right).   \] 
This correspondence breaks down however at the multi-variate level.
For instance, the following implication only holds in one direction with the converse requiring more sophisticated circuit-theoretic techniques to give a faithful compositional representation \cite{lorenz_unitary, vanrietvelde2025causaldecompositions1dquantum, Vanrietvelde_2021}:
 \[  \left(\tikzfigscale{1}{figs/caus_comp_3} \ = \ \tikzfigscale{1}{figs/caus_comp_7}\right)  \  \implies \  \left( \tikzfigscale{1}{figs/caus_comp_1}  \ = \   \tikzfigscale{1}{figs/caus_comp_2} \right) \ \wedge \ \ \left( \tikzfigscale{1}{figs/caus_comp_5} \ = \   \tikzfigscale{1}{figs/caus_comp_6} \right)  .   \] 

Secondly, in the study of higher-order processes in quantum theory \cite{taranto2025higherorderquantumoperations}, it is common to refer to diagrams such as  \[ \tikzfigscale{1}{figs/vertical_physics_1} \quad \quad \textrm{and} \quad \quad    \tikzfigscale{1}{figs/horizontal_physics_1},  \] as representing definite causal structures \cite{chiribella_circuits, oreshkov} and indefinite causal structures \cite{chiribella_switch, oreshkov} respectively, while canonical examples of such higher-order maps such as quantum combs and quantum switches can be thought of as definite and indefinite compositional structures.
The study of causal decompositions and higher-order processes are intimately related: single-input higher-order processes are equivalent to lower-order processes with one-way signalling constraints \cite{chiribella_supermaps, kissinger_caus}; the composition rules for higher-order processes can be characterised in terms of lower-order signalling constraints \cite{apadula2022nosignalling, simmons2024completelogiccausalconsistency}; lower-order processes satisfying single and multi-variate causality constraints are typically taken as the domains of those higher order processes which exhibit definite and indefinite causal structure respectively \cite{kissinger_caus, SimmonsKissinger2022,  simmons2024completelogiccausalconsistency, Bisio_2019, chiribella_switch, Wilson_causal, wilson2023mathematical, jencova2024structurehigherorderquantum, hoffreumon2022projective}; and the same more sophisticated circuit-theoretic techniques are needed to faithfully represent the known unitary indefinite causal structures in a purely compositional way \cite{Vanrietvelde_2025, grothus2025routingquantumcontrolcausal, mejdoub2025unilateraldeterminationcausalorder}.  
 
The purely categorical supermaps approach \cite{hefford_coend, wilson_locality, wilson_polycategories, Hefford_2024, hefford2025bvcategoryspacetimeinterventions} uses the theory of strong profunctors \cite{tambara, pastro_street} to define supermaps in terms of sequential and parallel composition, and so without direct reference to ideas from causality or leveraging categorical features of finite dimensional quantum systems such as compact closure \cite{selinger_cpm}. This is in contrast to the construction of higher-order causal categories \cite{kissinger_caus, simmons2024completelogiccausalconsistency, SimmonsKissinger2022} which directly uses causal and finite dimensional features of quantum theory to build a theory of higher-order quantum objects and their composition rules. 

In previous works, it has been demonstrated that the profunctorial approach recovers the definition of higher-order operation between a broad class of objects \cite{hefford_coend, wilson_locality, wilson_polycategories, Hefford_2024, hefford2025bvcategoryspacetimeinterventions} in a broad range of theories \cite{wilson2026supermapsgeneralisedtheories}. However, what has not been established is the extent to which the type-theory (the logical combinators for higher-order types) given by the profunctorial approach recovers the already well established type theory for higher-order quantum operations between general higher-order quantum objects \cite{kissinger_caus, simmons2024completelogiccausalconsistency, SimmonsKissinger2022, Bisio_2019, jencova2024structurehigherorderquantum, hoffreumon2022projective, Milz_2024}.

In this article we establish a precise relationship between higher-order quantum theory as it is traditionally defined, and the higher-order quantum theory which arises from applying the profunctorial approach to first-order quantum theory. In order to frame the comparison we use the language of categories with more than one tensor product - representing the ways in which systems can be spatio-temporally arranged. More precisely, we refer to the duoidal structure of higher-order processes which makes use of a tensor $\otimes$ and a sequencer $\seq$ representing space-like and time-like separated systems.
For any precausal category $\mathbf{C}$ we construct a functor $\mathcal{F}: \caus(\C) \rightarrow \StProf(\C_1)$ where $\mathbf{C}_1$ is the category of first-order causal processes in $\mathbf{C}$ such that $\mathcal{F}$ is:
\begin{itemize}
  \item injective on objects, 
  \item faithful,
  \item lax on the tensor $\mathcal{F}a \otimes \mathcal{F}b \rightarrow \mathcal{F}(a \otimes b)$, representing the coarse-graining from multi-variate compositionality constraints to multi-variate causality constraints,
  \item lax on the sequencer $\mathcal{F}a\seq \mathcal{F}b \morph{} \mathcal{F}(a \seq b)$. 
\end{itemize}
When the precausal category is additive as defined in \cite{simmons2024completelogiccausalconsistency, SimmonsKissinger2022}, the functor $\mathcal{F}$ is furthermore:
\begin{itemize}
  \item full, 
  \item strongly closed $\mathcal{F}[a,b] \cong [\mathcal{F}a, \mathcal{F}b]$, representing the fact that the entire higher-order structure of higher-order quantum processes lives inside of $\StProf(\C_1)$.
\end{itemize}
Finally, when $\mathbf{C}=\mathbf{CP}$ we find furthermore that $\mathcal{F}$ is strong on the sequencer $\mathcal{F}a\seq \mathcal{F}b \cong\mathcal{F}(a \seq b)$, representing a generalisation of the single-variance theorem which identifies semi-localizabiltiy with semi-causality. 
The laxity for the tensor and strength on the sequencer has a natural interpretation: it encodes the fact that there exists a causal decomposition for one-way signalling processes, but not for non-signalling processes. 
As a result, insofar as compositional constraints can be used to encode causality constraints, the structure of types for strong profunctors recovers the structure for types of higher-order quantum objects. 


\section{Higher-order quantum operations and category theory}



\subsection{Duoidal and BV-categories}

\begin{definition}
  A category $\cat{C}$ is \textit{duoidal} when it is equipped with two tensors $(\otimes,I_\otimes)$ and $(\seq,I_\seq)$ together with natural transformations,
  \begin{gather}\label{eq:duoidal}
    (A\seq B) \otimes (C\seq D) \morph{\delta_{ABCD}} (A\otimes C)\seq (B\otimes D), \qquad I_\otimes \morph{\gamma} I_\otimes\seq I_\otimes, \qquad I_\seq \otimes I_\seq \morph{\mu} I_\seq, \qquad I_\otimes \morph{\nu} I_\seq,
  \end{gather}
  satisfying a series of coherence conditions which can be found in e.g.\ \cite{aguiar}.
  A duoidal category is \textit{normal} when $I_\otimes \cong I_\seq$.
\end{definition}

\begin{definition}
  Let $\cat{C}$ and $\cat{D}$ be duoidal categories.
  A functor $F:\cat{C}\morph{}\cat{D}$ is \textit{lax-lax duoidal} when it is laxly monoidal in both tensor products,
  \begin{equation}\label{eq:duoidal_functor}
    FA \otimes FB \morph{} F(A\otimes B), \qquad I_\otimes \morph{} FI_\otimes, \qquad FA \seq FB \morph{} F(A\seq B), \qquad I_\seq \morph{} FI_\seq,
  \end{equation}
  satisfying the usual coherence conditions, plus additional compatibility with the structural cells \eqref{eq:duoidal}, see \cite{aguiar}.
  We will use terms like \textit{lax-strong} and \textit{strong-strong} to specify that $F$ is strong in one or more of the tensors so that the corresponding distributors of \eqref{eq:duoidal_functor} are isomorphisms.
\end{definition}

\begin{definition}
  A normal duoidal category $\cat{C}$ is a \textit{BV-category} if it equipped with a full and faithful functor $(-)^*:\opcat{\cat{C}}\morph{}\cat{C}$ inducing a natural isomorphism $\cat{C}(A\otimes B,C^*) \cong \cat{C}(A,(B\otimes C)^*)$ and such that there are natural isomorphisms $(A\seq B)^* \cong A^*\seq B^*$ and $I_\seq^*\cong I_\seq$.
  These data must satisfy a number of coherence conditions see \cite{blute,hefford2025bvcategoryspacetimeinterventions}.
\end{definition}
BV-categories are also by the above definition $*$-autonomous, meaning they can be equipped with a third monoidal product referred to as the par $\parr$, a dual to the already available tensor $\otimes$ in the sense that $A \parr B =  (A^* \otimes B^* )^*$.


\subsection{Higher-order causal categories}
Higher-order causal categories \cite{kissinger_caus, simmons2024completelogiccausalconsistency, SimmonsKissinger2022} generalise the construction of higher-order quantum theory ($\mathbf{HOQT}$) \cite{Bisio_2019, apadula2022nosignalling} to a construction which applies to precausal categories, that is, categories which admit both a causality interpretation and compact closure - a signature of finite-dimensionality.
In any precausal category $\mathbf{C}$ and for any subset $c \subseteq \mathbf{C}(I,A)$ we define the dual set \[ c^* = \left\{ \rho \in A^* \ s.t \ \forall \sigma \in c  :  \  \tikzfigscale{1}{figs/closure_of_state}  = 1 \right\} . \] 
We say that a set $c$ is \textit{first-order}, if $c^*$ is a singleton.
This notion of dual can then be used within a variant of double glueing to define the higher-order quantum operations, and more generally build a theory of higher-order causal processes from any precausal category. 
\begin{definition}
A \textit{precausal category} is a compact closed category $\mathbf{C}$ equipped with an environment structure, that is, for each object $A$ an associated preferred discarding effect $\tikzfigscale{1}{figs/ground_no}_A$ such that $\tikzfigscale{1}{figs/ground_no}_{A\otimes B} = \tikzfigscale{1}{figs/ground_no}_A \otimes \tikzfigscale{1}{figs/ground_no}_B$ and $\tikzfigscale{1}{figs/ground_no}_I = 1$.
Precausal categories also satisfy additional properties which we will not state explicitly here as they will not be referred to within this article, for full definitions the reader is referred to \cite{kissinger_caus, simmons2024completelogiccausalconsistency, SimmonsKissinger2022}:
\begin{itemize}
\item Enough causal states (an adaptation of well-pointed-ness) ,
\item$2^\text{nd}$-order causal processes factorise (an adaptation of the decomposition theorem for $1$-input superchannels \cite{chiribella_supermaps}) .
\end{itemize}
\end{definition}
\begin{definition}
  The \textit{Caus-construction} $\mathbf{Caus}(\bullet)$ sends any precausal category $\mathbf{C}$ to a BV-category $\mathbf{Caus}(\mathbf{C})$ with,
  \begin{itemize}
    \item Objects given by pairs $(c,A)$ with $A$ an object of $\mathbf{C}$ and $c \subseteq \mathbf{C}(I,A)$, where $c$ is a set that is \textit{flat} and \textit{closed}, meaning that:
    \begin{itemize}
      \item there exists an invertible scalar $\lambda\in\mathbb{C}$ such that $\left( \lambda\tikzfigscale{1}{figs/flat}  \in c \right)$,
      \item $c^{**} = c$.
    \end{itemize}
    \item Morphisms $(c,A) \rightarrow (c',A')$ given by morphisms $f : A \rightarrow A'$ such that $\forall \rho \in c : \ f \circ \rho \in c'$. 
    \item Dual $(c,A)^* = (c^* , A^*)$.
    \item Tensor product $(c_1,A_1) \otimes (c_2,A_2)  := ((c_1 \times c_2)^{**} , A_1 \otimes A_2)$.
    \item Par product $(c_1,A_1) \parr (c_2,A_2)  := (c_1 \parr c_2 ,  A_1 \otimes A_2)$ where $c_1 \parr c_2 = (c_1^* \times c_2^*)^{*}$.
    \item Sequencing Product $(c_1,A_1) \seq (c_2,A_2)  :=  (c_1\seq c_2, A_1\otimes A_2)$ where
    \[c_1\seq c_2 := \left\{  \tikzfigscale{1}{figs/seq_def_1}  \in (c_1 \parr c_2) \ : \ (c_Z \textrm{ is first-order})\ \wedge ( \rho \in c_1 \parr c_Z ) \wedge (\sigma \in c_Z^* \parr c_2 )  \right\}  .\]
    \item Internal Hom $[(c_1,A_1) , (c_2,A_2)] =  (c_1^* \parr c_2, A_1^* \otimes A_2)$.
  \end{itemize}
\end{definition}
Note that deterministic higher-order quantum theory as it is put forward in \cite{Bisio_2019, apadula2022nosignalling}, arises from keeping only those objects generated by first-order objects and the internal hom. 
\begin{remark}
We will from now on refer to objects of the category $\mathbf{Caus}(\mathbf{C})$ simply by capital symbols $A,B, \dots$. We refer to the set associated to object $A$ as $c_A$.
\end{remark}
Particularly important for our constructions are the first-order objects. 
\begin{definition}
The \textit{first-order objects}, are those objects $X$ within $\mathbf{Caus}(\mathbf{C})$ such that $c_X^*$ is a singleton.
We write $\mathbf{C}_1$ for the full subcategory of $\mathbf{Caus}(\mathbf{C})$ spanned by first-order objects.
When $\mathbf{C} = \mathbf{CP}$, we have that $\mathbf{C}_1 = \mathbf{CPTP}$. 
\end{definition}
\begin{remark}
First-order objects have a number of important properties, if $X$ is first-order, we have that $A \seq X = A \parr X$ for any object $A$. Furthermore, if both $X$ and $Y$ are first-order, then $X \otimes Y = X \seq Y = X \parr Y$. 
\end{remark}

The internal hom $[A,A']$ represents the space of all legitimate transformations from $A$ to $A'$. If $A,A'$ are first-order, then $[A,A']$ represents the set of channels with input $A$ and output $A'$. If $B,B'$ are also first-order, then the type $[[A,A'],[B,B']]$ represents the space of all superchannels which accept channels of type $[A,A']$ and return channels of type $[B,B']$. 

The tensor product represents the no-signalling constraint, for instance, given first-order objects $A,A',B,B'$ the space $[A, A'] \otimes [B,B']$ represents the space of non-signalling channels of type $[A \otimes B, A' \otimes B']$: \[ \phi \in  [A, A'] \otimes [B,B']  \ \iff \  \left( \tikzfigscale{1}{figs/caus_comp_1}  \ = \   \tikzfigscale{1}{figs/caus_comp_2} \right) \ \wedge \ \ \left( \tikzfigscale{1}{figs/caus_comp_5} \ = \   \tikzfigscale{1}{figs/caus_comp_6} \right)   \]
The sequencer, on the other hand, represents a one-way signalling constraint. Given first-order objects $A,A',B,B'$ the space $[A, A'] \seq [B,B']$ represents the space of one-way-signalling channels of type $[A \otimes B, A' \otimes B']$: \[  \phi \in [A, A'] \seq [B,B']  \ \iff \   \left(\tikzfigscale{1}{figs/caus_comp_1} \ = \ \tikzfigscale{1}{figs/caus_comp_2} \right) \ \iff \  \left(\tikzfigscale{1}{figs/caus_comp_3} \ = \   \tikzfigscale{1}{figs/caus_comp_4}\right)  \]
Finally, the par permits all conceivable signalling. More precisely, given first-order objects $A,A',B,B'$ the space $[A, A'] \parr [B,B']$ represents the space of all channels of type $[A \otimes B, A' \otimes B']$. 

\subsection{Strong profunctors}
Our other model of higher-order quantum theory comes from a series of works \cite{wilson_locality,wilson_polycategories,Hefford_2024,hefford2025bvcategoryspacetimeinterventions} using the theory of profunctors.
Let us now introduce the basic notions that will be required to understand this approach.

\begin{definition}
  A \textit{profunctor} $P:\cat{C}\pmorph\cat{C}$ is a functor $P:\opcat{\cat{C}}\times\cat{C}\morph{}\set$.
\end{definition}
We think of a profunctor as a generalised space of processes over $\cat{C}$.
These spaces of processes vary both covariantly and contravariantly over $\cat{C}$ modelling the idea that their inputs and outputs can be probed by $\cat{C}$.
For instance, $P(A,B)$ is a generalised space of processes $A\rightsquigarrow B$ and we can act on such a process by a morphism $B\morph{}C$ of $\cat{C}$ to get another generalised process $A\overset{p}{\rightsquigarrow} B\morph{f}C = P(1,f)(p) \in P(A,C)$.

Profunctors compose by an operation known as a \textit{coend} \cite{loregian_coend},
\begin{equation*}
  (P\seq Q)(A,B) = \int^X P(A,X)\times Q(X,B),
\end{equation*}
which one can think of as analogous to the matrix multiplication $(PQ)_A^B = \sum_X P_A^X Q^B_X$.
The elements of the set $(P\seq Q)(A,B)$ are given by pairs $(p,q)\in P(A,X)\times Q(X,B)$ quotiented by the equivalence relation generated by $(A\overset{p}{\rightsquigarrow} X\morph{f}X', X'\overset{q}{\rightsquigarrow} B) \sim (A\overset{p}{\rightsquigarrow}X, X\morph{f}X'\overset{q}{\rightsquigarrow} B)$ known as the \textit{sliding relations}.
In this sense the coend acts like the tensor product of bimodules.

\begin{definition}
  Let $\cat{C}$ be a monoidal category.
  A profunctor $P:\cat{C}\pmorph\cat{C}$ is \textit{strong} when it is equipped with a natural transformation $\zeta:P(A,B)\times\cat{C}(C,D) \morph{} P(A\otimes C,B\otimes D)$ satisfying coherence conditions making $\zeta$ associative and unital \cite{tambara,pastro_street}.
  A natural transformation $\eta:P\morph{}Q$ between strong profunctors is \textit{strong} when it commutes with the strengths.
\end{definition}

The strong profunctors and strong natural transformations assemble into a category $\StProf(\cat{C})$.
This category is $\otimes$-symmetric normal duoidal \cite{garner,earnshaw} when equipped with two tensors that behave like the tensor product of bimodules in the direction of composition and in the direction of the tensor product of $\cat{C}$.
The former is given by the coend $P\seq Q$ equipped with a canonical strength \cite{pastro_street}.
The latter is given by a quotient of the Day convolution \cite{day_thesis,day}, $\int^{ABCD} \cat{C}(-,A\otimes B) \times P(A,C) \times Q(B,D) \times \cat{C}(C\otimes D,\bl)$, over the category $\opcat{\cat{C}}\times\cat{C}$.
This quotient acts to coequalise the strengths of $P$ and $Q$ much like the tensor product of bimodules, and can be written succinctly as Day convolution over a certain category of \textit{coend optics},
\begin{equation*}
  (P\otimes Q)(-,\bl) = \int^{(A,A'),(B,B')}  P(A,A') \times Q(B,B') \times \opt(\cat{C})((A,A')\otimes(B,B'),(-,\bl)).
\end{equation*}
Both of these tensors have the identity profunctor $1=\cat{C}(-,\bl)$ as their unit object and both are closed with the internal hom for $\otimes$ given by,
\begin{equation*}
  [P,Q]_{\otimes}(-,\bl) = \int_{(A,A')\in\opt(\cat{C})} \set \left( P(A,A'), Q(A\otimes-,A'\otimes\bl) \right).
\end{equation*}


\section{A Lax Duoidal Embedding}
In this section we will prove that for any precausal category, there is a faithful duoidal embedding from $\caus(\C)$ to the category $\StProf(\mathbf{C}_1)$ of strong profunctors over the category of first-order causal systems. 

\begin{proposition}
  The following assignment,
  \begin{itemize}
    \item $\mathcal{F}(A)(X,X') = \caus(\C)(X, A \parr X')$,
    \item $\mathcal{F}(f:A \rightarrow B)(X,X') = \caus(\C)(X, f \parr X')$,
  \end{itemize}
  defines a $\otimes$-lax monoidal functor $\mathcal{F}:\caus(\C) \rightarrow \StProf(\C_1)$.
  Note that $X$ and $X'$ vary over \textit{only} the first-order objects, i.e.\ those of $\C_1$.
\end{proposition}

We will for readability express where possible algebraic/categorical definitions graphically. 
We may draw the strong profunctor $\mathcal{F}(A)$ as 
\[ \mathcal{F}(A)(X,X') := \left\{ \tikzfigscale{1}{figs/diag_str_1} \right\} , \] with action on morphisms, and strength given by \[   \mathcal{F}(A)(g,h) : \tikzfigscale{1}{figs/diag_str_1}  \mapsto   \tikzfigscale{1}{figs/diag_str_2} , \quad \quad  \zeta_{XX'YY'} :\left(  \tikzfigscale{1}{figs/diag_str_1} , \  \tikzfigscale{1}{figs/diag_tuple}  \right)  \mapsto   \tikzfigscale{1}{figs/diag_str_3}. \]
As the functor $\mathcal{F}$ sends each object $A$ to a profunctor $\mathcal{F}(A)(-,=)$, each morphism $f:A \rightarrow B$ must be sent to a strong natural transformation $\mathcal{F}(f)_{X,X'} : \mathcal{F}(A)(X,X') \rightarrow \mathcal{F}(B)(X,X')$. In this case, we define this strong natural transformation by: \[ \mathcal{F}(f)_{XX'} : \tikzfigscale{1}{figs/diag_str_1}  \mapsto   \tikzfigscale{1}{figs/diag_str_4}.  \]
Naturality and strength of the family of functions $\mathcal{F}(f)_{XX'} $ then follows by the unambiguous interpretation of \[  \tikzfigscale{1}{figs/diag_str_6}  \quad \textrm{and} \quad   \tikzfigscale{1}{figs/diag_str_7} ,\] respectively. 
Note, that $\mathcal{F}(A)(X,X') = \caus(\C)(X, A \parr X') \cong  \caus(\C)(I_{\caus(\C)}, A \parr [X, X'])$, where then the existence of this functor as a multifunctor can be understood as arising from the fact that $\caus(\C)$ is a higher-order circuit theory \cite{wilson2026higherordercircuits}.
\begin{proof}
  Note that whenever $X$ is a first-order object $A\seq X = A\parr X$ for any $A$ so that the image of the functor $\mathcal{F}$ can equally be written $\caus(\C)(X,A\parr X') = \caus(\C)(X,A\seq X')$.
  The assignment is clearly functorial so we deal just with the laxators, which are given by,
  \begin{align*}
    \caus(\C)(X,A\seq X') \times \caus(\C)(Y,B\seq Y') & \morph{\otimes} \caus(\C)(X\otimes Y, (A\seq X')\otimes(B\seq Y')) \\
    & \morph{\caus(\C)(1,\delta)} \caus(\C)(X\otimes Y,(A\otimes B)\seq(X'\otimes Y'))
  \end{align*}
  while the unit distributor is given by,
  \begin{equation*}
    \C_1(-,\bl) \cong \caus(\C)(-,\bl) \cong \caus(\C)(-,I\parr\bl)
  \end{equation*}
  since the morphisms in $\caus(\C)$ between first-order objects are precisely those of $\C_1$.
\end{proof}
The preceding proof makes direct use of the fact that to give a strong natural from the tensor $\eta_{XX'} :  (P\otimes Q)(X,X') \rightarrow R(X,X')$ is to give a multi-strong natural transformation from the product $\eta : P(X_p,X_p') \times Q(X_q,X_q') \rightarrow R(X_p \otimes X_q , X_p' \otimes X_q')$. At the level of the product, the laxator simply acts as  \[   \otimes_{X_1 X_1' X_2 X_2'} : \left(  \tikzfigscale{1}{figs/diag_lax_1}  , \   \tikzfigscale{1}{figs/diag_lax_2} \right)  \mapsto   \tikzfigscale{1}{figs/diag_lax_3} . \] 

Before proving that the functor is (lax) monoidal on the sequencer, let us explicitly describe the behaviour of the sequencer of strong profunctors for images of the functor $\mathcal{F}$. 
For any two objects $A,B$ of $\mathbf{Caus}(\mathbf{C})$ the elements of the sequencing tensor product \[  \int^Z \mathcal{F}(A)(-,Z) \times \mathcal{F}(B)(Z,\bl)  \]
are the elements of the cartesian product  \[ \mathbf{Caus}(\mathbf{C})(-, A \parr Z) \times  \mathbf{Caus}(\mathbf{C})(Z, B \parr \bl), \] up to equivalence by transitive closure of the relation \[   \tikzfigscale{1}{figs/new_optic_1} \ \cong \  \tikzfigscale{1}{figs/new_optic_2}  . \]
We must take great care when using this deceivingly nice graphical notation to establish what can be slid along such wires. Any first-order causal morphism, that is, any morphism from $\mathbf{C}_1$, can be slid.
\begin{theorem}
  The functor $\mathcal{F}:\caus(\C) \rightarrow \StProf(\C_1)$ is lax on the sequencing tensor $\seq$, that is there are distributors,
  \[\mathcal{F}(A) \seq \mathcal{F}(B) \morph{} \mathcal{F}(A \seq B), \qquad I\morph{}\mathcal{F}I, \] 
  natural in $A$ and $B$ and satisfying the required coherences.
\end{theorem}
\begin{proof}
  Note that whenever $X$ is a first-order object $A\seq X = A\parr X$ for any $A$ so that the image of the functor $\mathcal{F}$ can equally be written $\caus(\C)(X,A\parr X') = \caus(\C)(X,A\seq X')$.
  The distributor is given by composition along $X$,
  \begin{equation*}
    FA \seq FB = \int^{X\in\C_1} \caus(\C)(-,A\seq X) \times \caus(\C)(X,B\seq \bl) \morph{} \caus(\C)(-,A\seq(B\seq\bl)) \cong F(A\seq B),
  \end{equation*}
  while the unit distributor is given by,
  \begin{equation*}
    \C_1(-,\bl) \cong \caus(\C)(-,\bl) \cong \caus(\C)(-,I\parr\bl)
  \end{equation*}
  since the morphisms in $\caus(\C)$ between first-order objects are precisely those of $\C_1$.
\end{proof}
Graphically, this laxator for the sequencer (which we will refer to as $\theta$) simply contracts a coend-composition into a literal compositional of diagrams \[  \theta_{X,X'} : \  \tikzfigscale{1}{figs/new_optic_0} \ \mapsto \ \tikzfigscale{1}{figs/back_strong_2}. \]

\begin{proposition}
  The functor $\mathcal{F}:\caus(\C) \rightarrow \StProf(\C_1)$ is injective on objects.
\end{proposition}
\begin{proof}
  Suppose $\mathcal{F}(A)(-,\bl) = \mathcal{F}(B)(-,\bl)$.
  Upon evaluating at trivial systems we have
  \begin{align*}
    \mathcal{F}(A)(I,I) = \mathcal{F}(B)(I,I) \implies \caus(\C)(I,A) = \caus(\C)(I,B),
  \end{align*}
  but an object of $\caus(\C)$ is completely determined by its set of states so $A=B$.
\end{proof}

\begin{proposition}
The functor $\mathcal{F}:\mathbf{Caus}(\mathbf{C}) \rightarrow \StProf(\mathbf{C}_1)$ is faithful.
\end{proposition}
\begin{proof}
Given in the Appendix. 
\end{proof}

\section{The Embedding is Full and Strongly Closed}
In this section we show that whenever a precausal category is additive, the embedding of the Caus construction into the category of strong profunctors is full and strongly closed. We refer the reader to \cite{simmons2024completelogiccausalconsistency, SimmonsKissinger2022} for a complete definition of additive precausal categories. Here, we note that such categories include amongst other conditions, the existence of all products (and so in this case additive enrichment), and the possibility to complete any effect to the discard-morphism up to a scalar. These points allow us to prove (given in the Appendix) that closed flat sets are closed under convex combinations, and support a notion of controlled process for any family of processes $\{ \rho_i \}$. 
\begin{proposition}
  Let $\mathbf{C}$ be an additive precausal category, then the functor $\mathcal{F}:\caus(\C) \rightarrow \StProf(\C_1)$ is full.
\end{proposition}
\begin{proof}
 To prove that $\mathcal{F}$ is full, amounts to a generalisation of the characterisation theorem of \cite{wilson_locality} which applies to convex normal sets of processes, to general types in additive higher-order causal categories. 
 The proof adopts a graphical notation in which a natural transformation $S_{XX'} : \mathcal{F}(A)(X,X') \rightarrow \mathcal{F}(B)(X,X')$ is depicted:  \[ \tikzfigscale{1}{figs/diag_net_1} , \] with strong naturality encoded graphically by  \[ \tikzfigscale{1}{figs/diag_nat_2} \ = \ \tikzfigscale{1}{figs/diag_nat_3} , \quad \quad \tikzfigscale{1}{figs/diag_nat_4} \ = \ \tikzfigscale{1}{figs/diag_nat_5} , \]
 and then takes three key steps, following the methods of \cite{wilson2026supermapsgeneralisedtheories, wilson_locality}:
 \begin{itemize}
 \item Proving convex-linearity of each function $S_{XX'}$,
 \item Proving what is referred to in \cite{wilson2026supermapsgeneralisedtheories} as non-contextual assignment, that for even non-causal $\pi_i \leq  \tikzfigscale{1}{figs/ground}$:  \[  \tikzfigscale{1}{figs/extension_1} \ = \   \tikzfigscale{1}{figs/extension_2} \implies \tikzfigscale{1}{figs/extension_3} \ = \   \tikzfigscale{1}{figs/extension_4},  \] 
 \item A Yoneda-like step, in which the candidate internal representation for $S_{XX'}$ is constructed by evaluating it on the cup $\cup$ (this is most similar structurally to the ansatz used in \cite{wilson_locality}) up to a scalar $\alpha_A$: \[   S'  \ :=   \ \tikzfigscale{1}{figs/preimage_1},  \] where $\alpha_A$ satisfies  \[ \tikzfigscale{1}{figs/cup_1} \in c_{A \parr |A|}. \]
 \end{itemize} 
 A full proof for completeness is given in the Appendix.
\end{proof}

\begin{proposition}
  Let $\mathbf{C}$ be an additive precausal category, then the functor $\mathcal{F}:\caus(\C) \rightarrow \StProf(\C_1)$ is strongly closed.
\end{proposition}
\begin{proof}
  For any pair of strong profunctors,
  \begin{equation*}
    [P,Q](X,X') = \int_{(A,A')} \set \left( P(A,A'), Q(A\otimes X,A'\otimes X') \right) \cong \StProf(\C_1) \left(P(-,\bl),Q(-\otimes X,\bl\otimes X') \right)
  \end{equation*}
and so we have that 
  \begin{align*}
    [\F A,\F B](X,X') & \cong \StProf(\C_1) \left(\caus(\C)(-,A\parr \bl), \caus(\C)(-\otimes X,B \parr (\bl\otimes X')) \right) \\
    & \cong \StProf(\C_1) \left(\caus(\C)(-,A\parr \bl), \caus(\C)(-,([X,B \parr X']) \parr (\bl )) \right) \\
        & \cong  \caus(\C)(A, [X,B \parr X']) \\
                & \cong  \caus(\C)(X, [A,B] \parr X') \\
                & \cong \F[ A,B](X,X').
  \end{align*}
\end{proof}

\section{The Embedding is Strong on the Sequencer}
Now we establish an even stronger property of the embedding, which we prove only for specifically higher-order quantum theory, based on the fact that combs are equal to optics in mixed quantum theory \cite{Hefford_2024}.
First, we will need a preliminary lemma. 
\begin{proposition}\label{prop:firstorder_existence}
  Let $\rho_1,\sigma_1, \rho_2',\sigma_2'$ be morphisms in $\mathbf{Caus}(\mathbf{CP})$ such that  \[    \tikzfigscale{1}{figs/iback_strong_2}  \ = \   \tikzfigscale{1}{figs/back_strong_3} ,\] with $X,X',Z_1,Z_2$ first-order objects.
  Then there exists corresponding purifications $P_1, \Sigma_1,P_2, \Sigma_2 $, and first-order causal processes $\pi$ and $v$ such that: 
  \[   \tikzfigscale{1}{figs/dilate_strong_5} \ = \      \tikzfigscale{1}{figs/dilate_strong_6} ,  \quad \quad \tikzfigscale{1}{figs/shadow_1} \ = \   \tikzfigscale{1}{figs/shadow_2},  \quad \quad   \tikzfigscale{1}{figs/shadow_3} \ = \   \tikzfigscale{1}{figs/shadow_4}.  \] 
\end{proposition}
\begin{proof}
  Given in the appendix.
\end{proof}

\begin{theorem}
  The functor $\mathcal{F}:\caus(\mathbf{CP}) \rightarrow \StProf(\mathbf{CPTP})$ is strong on the sequencing tensor $\seq$ so that for each pair of objects $A,B$,
  \[\mathcal{F}A \seq \mathcal{F}B \cong \mathcal{F}(A \seq B).\] 
\end{theorem}
\begin{proof}
  We must construct an inverse strong natural transformation $ \theta^{-1}: \mathcal{F}(A \seq B) \rightarrow \mathcal{F}A \seq \mathcal{F}B $. For this inverse arrow note that we can construct the embedding
  \begin{align*}
     & [X , (A \seq B) \parr X'] \\
    \cong & (A \seq B) \parr [X,X'] \\
    \cong &  ( A \seq B ) \parr ([X,I] \seq [I,X']) \\
     \rightarrow & (A \parr [X,I]) \seq (B \parr [I,X']), \\ 
          \cong & [X,A] \seq (B \parr X'), \\ 
  \end{align*}
  and so we have that for each $X,X'$ in $\mathbf{CPTP}$, then $ \tau \in \caus(\mathbf{CP})(X, (A \seq B) \parr X') $ implies that there exists a first-order object $Z$ such that 
  \begin{align*}
   \tikzfigscale{1}{figs/back_strong_1}  \ = \  \tikzfigscale{1}{figs/back_strong_2}  .
  \end{align*}
  Consequently we may simply choose the inverse $\theta^{-1}_{X,X'}(\tau) := (\rho, \sigma)$,
  however for this choice to be well-formed (and recalling that the sets being acted on are quotiented sets of states) we need to verify that for all $\rho,\sigma \in c_{[X, A \parr Z]} \times c_{[Z, B \parr X']}$ then:
  \[    \tikzfigscale{1}{figs/iback_strong_2}  \ = \   \tikzfigscale{1}{figs/back_strong_3} \ \implies \ (\rho_1,\sigma_1) \cong_{\texttt{co-end}} (\rho_2, \sigma_2).    \] 
  Recall that two elements of the coend are equivalent if they can be reached from one another by sliding only morphisms of $\mathbf{CPTP}$ and not the larger category $\mathbf{CP}$.
  By Proposition \ref{prop:firstorder_existence} all of the following slides are of first-order morphisms and the final proof proceeds as follows.  
  \begin{align*}
      \tikzfigscale{1}{figs/strong_coend_1} \ & = \   \tikzfigscale{1}{figs/strong_coend_2} \ = \  \tikzfigscale{1}{figs/strong_coend_3}   
    \ \sim \   \tikzfigscale{1}{figs/strong_coend_4} \ 
    = \  \tikzfigscale{1}{figs/strong_coend_7b} \ \sim \ \tikzfigscale{1}{figs/strong_coend_7} \\ 
    &   = \  \tikzfigscale{1}{figs/strong_coend_8} \   \sim \   \tikzfigscale{1}{figs/strong_coend_9} \  
    = \    \tikzfigscale{1}{figs/strong_coend_10} \ = \  \tikzfigscale{1}{figs/strong_coend_11} .   
  \end{align*}
\end{proof}

\section{Conclusion}

There are a few natural directions to take the results of this paper. 
By the Yoneda lemma for categorical supermaps \cite{wilson2026supermapsgeneralisedtheories}, natural transformations between strong profunctors on CPTs \cite{gogioso_cpt} of certain types are totally characterised in terms of the Choi-Jamiołkowski isomorphism when present, leading naturally to the question of whether the results of this paper extend to CPTs. This would give more concrete proposals for general higher-order objects within a variety of theories for which classes of higher-order processes have been introduced such as Boxworld \cite{bavaresco2024indefinitecausalorderboxworld}, Hex-square theory \cite{sengupta2024achievingmaximalcausalindefiniteness}, and time-symmetric quantum theory \cite{mrini2024indefinitecausalstructurecausal}, along with the possibility to compare with higher-order theory for bidirectional processes \cite{apadula2026higherordertransformationsbidirectionalquantum}. 

While $\StProf(\cat{C}_1)$ is typically not a BV-category, it can be lifted to one by the Chu construction \cite{chu}.
This standard construction from linear logic produces $\ast$-autonomous categories from closed monoidal ones and when applied to $\otimes$-closed normal duoidal categories it produces BV-categories \cite{hefford2025bvcategoryspacetimeinterventions}.
Applying this to strong profunctors produces a BV-category $\StEnv(\C_1)= \Chu(\StProf(\C_1))$ referred to as the \textit{strong Hyland envelope} due to its connections to the Hyland envelope of \cite{hyland_envelope,shulman_envelope}. 
It is natural then to wonder whether for any additive precausal category $\mathbf{C}$, the fully faithful lax-lax duoidal embedding of this article lifts to give a full sub-category of the strong Hyland envelope which is lax-lax BV-equivalent to the caus-construction. 

Finally, building upon this, the broadest possible sets of compositional rules available for higher-order quantum objects have been studied in \cite{simmons2024completelogiccausalconsistency, apadula2022nosignalling}. This suggests investigating the question of which (if any) compositional rules are unique to higher-order quantum operations, and which are in fact a feature of any reasonable theory of higher-order processes?
A study of the available compositional rules within the strong Hyland envelope over both general monoidal categories and its specialisation to action on CPTs \cite{gogioso_cpt} and OPTs \cite{chiribella_purification}, seems likely to give a concrete route towards answering this foundational problem on the compositionality of higher-order processes.

\subsection*{Acknowledgements}
    James Hefford is funded by the French National Research Agency (ANR) within the framework of ``Plan France 2030'', under the research projects EPIQ ANR-22-PETQ-0007 and HQI-R\&D ANR-22-PNCQ-0002. MW was funded by the Engineering and Physical Sciences Research Council [grant number EP/W524335/1].

\bibliographystyle{plainurl}
\bibliography{bibliography}

\appendix

\section{Faithfulness}

\begin{proposition}
The functor $\mathcal{F}:\mathbf{Caus}(\mathbf{C}) \rightarrow \StProf(\mathbf{C}_1)$ is faithful.
\end{proposition}
\begin{proof}
Consider $f:A \rightarrow B$ and $g: A \rightarrow B$ such that $\mathcal{F}(f) = \mathcal{F}(g)$. Note that for any object $A$ we can define $|A|$ as the object of $\mathbf{Caus}(\mathbf{C})$ with the same underlying object as $A$ but with \textit{all causal states}. Then see that for any object $A$  there exists an invertible scalar $\alpha_A$ such that \[ \tikzfigscale{1}{figs/cup_1} \in c_{A \parr |A|}, \] precisely, one simply chooses the $\alpha$ such that  \[ \tikzfigscale{1}{figs/cup_1b} \in c_A  \] and then notes that for this $\alpha$ and for any $\pi \in (c_A)^{*}$,  \[ \tikzfigscale{1}{figs/cup_2}  = 1 . \] 
Finally, noting that $A \parr |A| \cong A \parr A [I,|A|]$ with $|A|$ a first-order object we have that
\begin{align*}
\mathcal{F}(f) = \mathcal{F}(g) & \implies \mathcal{F}(f)_{I,|A|} = \mathcal{F}(g)_{I,|A|} \\
& \implies \caus(\C)(I, f \parr [I,|A|]) = \caus(\C)(I, g \parr [I,|A|]) \\
& \implies  \caus(\C)(I, f \parr [I,|A|])(\alpha \cup) = \caus(\C)(I, g \parr [I,|A|])(\alpha \cup) \\
& \implies  \tikzfigscale{1}{figs/cup_3} \ = \  \tikzfigscale{1}{figs/cup_4} \\
& \implies  \tikzfigscale{1}{figs/cup_5} \ = \  \tikzfigscale{1}{figs/cup_6} \\
& \implies  \tikzfigscale{1}{figs/cup_7} \ = \  \tikzfigscale{1}{figs/cup_8} 
\end{align*}
\end{proof}

\section{Fullness of the embedding for additive precausal categories}
\begin{proposition}
Let $\mathbf{C}$ be an additive precausal category, then each $c_A$ is convex (closed under convex combinations of scalars in the precausal category).
\end{proposition}
\begin{proof}
Recall that each $c_A$ defined through normalisation to $1$ under the effect space $c_A^*$. Indeed let $\pi \circ \rho_i = 1$ for $\pi \in c_A^*$ and $\rho_i \in c_A$ then it follows that (using additive enrichment) $\pi \circ (\sum_i \rho_i) = \sum_i p_i \pi \circ \rho_i = \sum_i p_i = 1$ for any discrete probability distribution $p_i$. 
\end{proof}
\begin{proposition}
For any collection of states $\{ \rho_i \}_{i \in 1 \dots N}$ there exists a morphism $\texttt{ctrl}\{ \rho_i \}$ of type $\oplus_{i = 1}^n I \rightarrow A$ such that $\texttt{ctrl}\{ \rho_i \} \left(  \tikzfigscale{1}{figs/ground_state_i}  \right)  = \rho_i$. 
\end{proposition}
Here, by $ \tikzfigscale{1}{figs/ground_state_i} : I \rightarrow \oplus_{i = 1}^n I$ we mean the canonical map for the $i^\text{th}$-component of the biproduct $\oplus_{i = 1}^n$.
\begin{proof}
For any collection of states $\{ \rho_i \}_{i \in 1 \dots N}$ of type $A$ we can construct (using the existence of biproducts) the desired morphism $\texttt{ctrl}\{ \rho_i \}$ of type $\oplus_{i = 1}^n I \rightarrow A$ as simply \[ \texttt{ctrl}\{ \rho_i \} \  = \ \sum_{i}  \left(\tikzfigscale{1}{figs/new_control_2} \right).    \] 
Indeed, we can confirm that for any $\tau \in  c_{\oplus_{i = 1}^n I}$,
\[   \sum_{i} \left( \tikzfigscale{1}{figs/new_control_3} \right) \ = \ \sum_{i} p(i | \tau)  \ \tikzfigscale{1}{figs/new_control_4} \ \in  \ c_A,  \]
so that the morphism is of the required type.
Finally, we confirm the control property.
\[   \sum_{i} \left( \tikzfigscale{1}{figs/new_control_5} \right) \ = \ \sum_{i} p(i | \tau)  \ \tikzfigscale{1}{figs/new_control_6} \ \in  \ c_A.  \]
\end{proof}

\begin{proposition}
  Let $\mathbf{C}$ be an additive precausal category, then the functor $\mathcal{F}:\caus(\C) \rightarrow \StProf(\C_1)$ is full.
\end{proposition}
\begin{proof}
First, we will show that for any morphism of strong profunctors $S: \mathcal{F}(A) \rightarrow \mathcal{F}(B)$ every component $S_{X,X'}$ is a convex linear function $S_{X,X'}: c_{[X,A \parr X']} \rightarrow c_{[X,B \parr X']}$.
To do so, consider some family of states $\rho_i$ of type $[X,A \parr X']$ and construct the controlled version of this family of type 
\begin{align*}
[\oplus_{i = 1}^n I, [X,A \parr X']]  \cong [(\oplus_{i = 1}^n I) \otimes X,A \parr X'].
\end{align*}
Now let us confirm convexity 
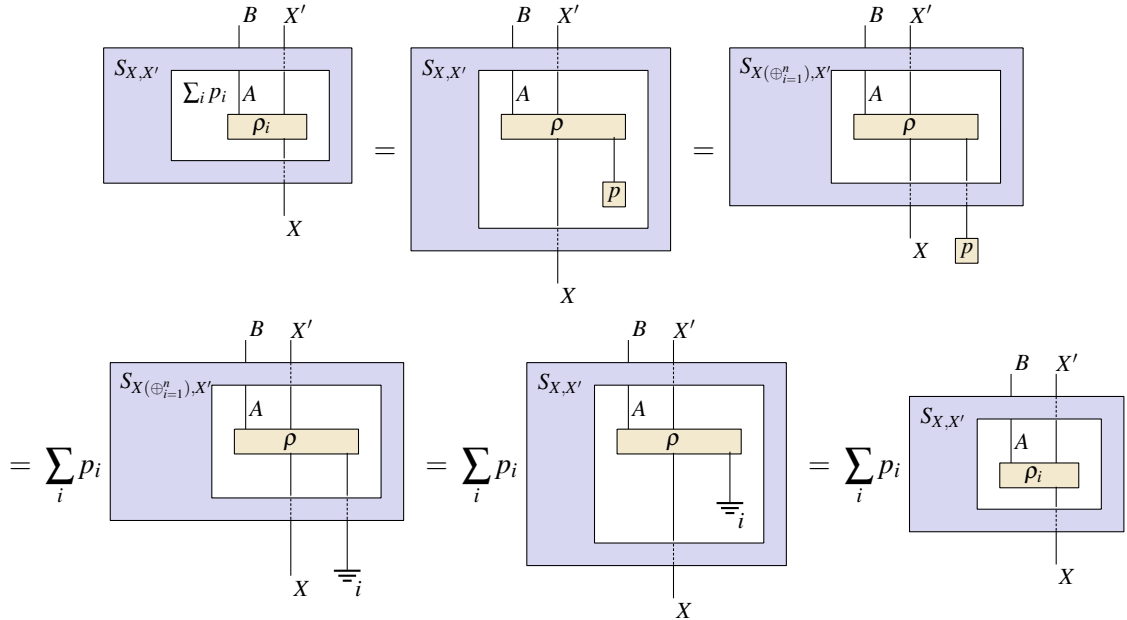
\[
    \tikzfigscale{1}{figs/iconvex_1} \ = \  \tikzfigscale{1}{figs/iconvex_2}  \ = \  \tikzfigscale{1}{figs/iconvex_3}   \]
   \[  \ = \ \sum_i p_i  \   \tikzfigscale{1}{figs/iconvex_6}  \ = \ \sum_i p_i  \   \tikzfigscale{1}{figs/iconvex_7}  \ = \ \sum_i p_i  \   \tikzfigscale{1}{figs/iconvex_8} \]

   Next we show a useful technical property of the $S_{X,X'}$ components for each $X,X'$.
   Let us imagine that there exists $\pi_i$, $ \rho_i$,  with each $\pi_i \leq  \tikzfigscale{1}{figs/ground}$ such that \[  \tikzfigscale{1}{figs/extension_1} \ = \   \tikzfigscale{1}{figs/extension_2},  \] 
   
   then it follows that 
   \[  \tikzfigscale{1}{figs/extension_3} \ = \   \tikzfigscale{1}{figs/extension_4}.  \] 
   To prove this note that we can construct from each such $\pi_i$ of type $X_i \rightarrow I$ in the underlying precausal category a new morphism $\Sigma_i: X_i \rightarrow (I \oplus I)$ in $\caus(\C)$ by  \[  \tikzfigscale{1}{figs/extension_proof_1}  \ \otimes \       \tikzfigscale{1}{figs/extension_less_1} \ + \ \tikzfigscale{1}{figs/extension_proof_2} \  \otimes  \ \left(  \  \tikzfigscale{1}{figs/iextension_less_2} \ - \  \tikzfigscale{1}{figs/extension_less_1} \right).  \] 
   Now we note that
            \begin{align*}
 \tikzfigscale{1}{figs/iget_zero_0a} \ - \  \tikzfigscale{1}{figs/iget_zero_0b}  \ = \  \tikzfigscale{1}{figs/iget_zero_1} \ - \  \tikzfigscale{1}{figs/iget_zero_2},
      \end{align*}
      from which we have that 
       \begin{align*}
\frac{1}{2} \tikzfigscale{1}{figs/iget_zero_0a} \ + \ \frac{1}{2} \tikzfigscale{1}{figs/iget_zero_2}  \ = \ \frac{1}{2} \tikzfigscale{1}{figs/iget_zero_1} \ + \ \frac{1}{2} \tikzfigscale{1}{figs/iget_zero_0b},
      \end{align*}
    Upon applying $S_{I,(I \oplus I)}$ to each side of the equation, using convex linearity, and then un-arranging we find that 
      \begin{align*}
 \tikzfigscale{1}{figs/get_zero_5} \ - \  \tikzfigscale{1}{figs/get_zero_6}  \ = \  \tikzfigscale{1}{figs/iget_zero_7} \ - \  \tikzfigscale{1}{figs/iget_zero_8} .
      \end{align*} Now we are ready to prove, 
      \begin{align*}
        \tikzfigscale{1}{figs/extension_3} \ - \   \tikzfigscale{1}{figs/extension_4} & \  = \       \tikzfigscale{1}{figs/iget_zero_9} \ - \   \tikzfigscale{1}{figs/iget_zero_10}    \\
  & \  = \       \tikzfigscale{1}{figs/iget_zero_11} \ - \   \tikzfigscale{1}{figs/iget_zero_12} \\
 & \  = \       \tikzfigscale{1}{figs/iget_zero_13} \ - \   \tikzfigscale{1}{figs/iget_zero_14} \\
  & \  = \       \tikzfigscale{1}{figs/get_zero_15} \ - \   \tikzfigscale{1}{figs/get_zero_16} \\
  & \ = \ 0 - 0 \ = \ 0
      \end{align*} 

  Finally, we are able to prove our main characterisation theorem, which establishes that $\mathcal{F}$ is full, by establishing that every possible morphism of strong profunctors $S:\mathcal{F}(A) \rightarrow \mathcal{F}(B)$ is in the image of $\mathcal{F}$, i.e.\ is of the form $\mathcal{F}(S')$.
  We must explicitly construct this pre-image $S': A \rightarrow B$ in $\caus(\C)$ and do so simply by defining \[   S'  \ :=   \ \tikzfigscale{1}{figs/preimage_1}.  \] Now we must confirm that indeed $\mathcal{F}(S') = S$, to do so note that
      \begin{align*}
     \mathcal{F}(S')(X,X')(\rho) \ =  \ \tikzfigscale{1}{figs/preimage_2} \ & = \ \tikzfigscale{1}{figs/preimage_3},
      \end{align*} and now since $\mathbf{C}$ is an additive precausal category for all $\pi$ there exists $\kappa$ 
      such that 
      $\kappa ( \cap \otimes \pi)  \leq   \tikzfigscale{1}{figs/ground_no} \otimes   \tikzfigscale{1}{figs/ground_no} \otimes  \tikzfigscale{1}{figs/ground_no}$. 
      This allows use to use the previous technical lemma to conclude that for all such $\pi$,
      \[   \tikzfigscale{1}{figs/preimage_5} \ = \   \tikzfigscale{1}{figs/preimage_6},   \]
      and since there are enough such $\pi$ to fix any state, we can conclude that
      \[   \tikzfigscale{1}{figs/preimage_3} \ = \   \tikzfigscale{1}{figs/preimage_4} \ = \   \tikzfigscale{1}{figs/preimage_4b}.   \] 
      Finally, we note that for all $\tau$
      \[   \tikzfigscale{1}{figs/final_0} \ = \   \tikzfigscale{1}{figs/final_1}  \ = \   \tikzfigscale{1}{figs/final_2},   \] 
      and so  \[   \tikzfigscale{1}{figs/final_3}  \ = \   \tikzfigscale{1}{figs/final_4},.  \]  

      This completes the proof that indeed $\mathcal{F}(S') = S$.  
\end{proof}

\section{Shadows for higher-order objects}

\begin{proposition}
  Consider $\caus(\mathbf{CP})$.
  For any $\rho$ of type $X \rightarrow A$ there exists a first order object $E_{\psi}$ and a pure morphism $\psi$ of type $X \rightarrow E_{\psi} \parr A$ such that \[ \tikzfigscale{1}{figs/dilate_1} \ = \ \tikzfigscale{1}{figs/dilate_2}.  \]
  Furthermore this purification is unique up to isometry, meaning that for any other purification $(E_{\psi'} , \psi')$ with $\texttt{dim}(E_{\psi}) \leq \texttt{dim}(E_{\psi'})$ there exists a pure (first-order causal) morphism of type $v: E_{\psi} \rightarrow E_{\psi'}$ such that   \[ \tikzfigscale{1}{figs/dilate_3} \ = \ \tikzfigscale{1}{figs/dilate_4}.  \]

\end{proposition}
\begin{proof}
  It is immediate by the dilation theorem for completely positive maps that such a $\psi$ satisfying the above equation exists.
  To check that is has type $X \rightarrow E_{\psi} \parr A$ is then clear since by the defining equation it is normalised under any element of $c_{E_{\psi}}^{*} \times c_A^{*}$. The existence of an isometry $v$ satisfying the required equation is directly inherited too, and since it is an isometry, it indeed is of type $v: E_{\psi} \rightarrow E_{\psi'}$ for $E_{\psi}$ and $E_{\psi'}$ first-order types. 
\end{proof}

\begin{proposition}
  Let $\rho_1,\sigma_1, \rho_2',\sigma_2'$ be morphisms in $\mathbf{Caus}(\mathbf{CP})$ such that  \[    \tikzfigscale{1}{figs/iback_strong_2}  \ = \   \tikzfigscale{1}{figs/back_strong_3} ,\] with $X,X',Z_1,Z_2$ first-order objects then there exists corresponding purifications $P_1, \Sigma_1,P_2, \Sigma_2 $, and first-order causal processes $\pi$ and $V$ such that: 
  \[   \tikzfigscale{1}{figs/dilate_strong_5} \ = \      \tikzfigscale{1}{figs/dilate_strong_6} ,  \quad \quad \tikzfigscale{1}{figs/shadow_1} \ = \   \tikzfigscale{1}{figs/shadow_2},  \quad \quad   \tikzfigscale{1}{figs/shadow_3} \ = \   \tikzfigscale{1}{figs/shadow_4}.  \] 
\end{proposition}
\begin{proof}
  Firstly, note that each of these states comes equipped with a corresponding purification $P,\Sigma,P',\Sigma'$ such that 
  \[ \tikzfigscale{1}{figs/dilate_strong_1} \ = \      \tikzfigscale{1}{figs/dilate_strong_2},   \] 
  from which we find for any morphism in $\mathbf{Caus}(\mathbf{CP})$ of type \[ \tikzfigscale{1}{figs/dilate_strong_new_1} \ = \      \tikzfigscale{1}{figs/dilate_strong_new_2},   \]  and so
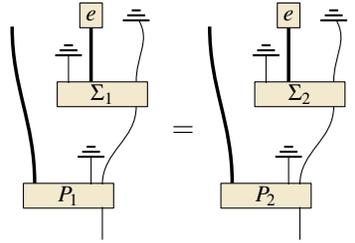
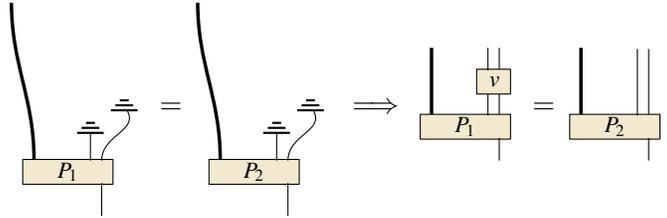
  \[ \tikzfigscale{1}{figs/dilate_strong_3} \ = \      \tikzfigscale{1}{figs/dilate_strong_4}  \ \implies \  \tikzfigscale{1}{figs/dilate_strong_5} \ = \      \tikzfigscale{1}{figs/dilate_strong_6} . \] 

  It now follows that,
  \[\tikzfigscale{1}{figs/dilate_strong_1} \ = \  \tikzfigscale{1}{figs/dilate_strong_2} \ = \      \tikzfigscale{1}{figs/dilate_strong_7}.   \] 
From here, the proof of \cite{carette_delayedtrace} on the existence of shadows for completely positive maps applies, and we are able to construct an idempotent CPTP (and so, first-order causal) map $\pi$ satisfying: 
   \[   \tikzfigscale{1}{figs/shadow_1n} \ = \   \tikzfigscale{1}{figs/shadow_2n} \quad \quad  \text{and} \quad \quad   \tikzfigscale{1}{figs/shadow_3n} \ = \   \tikzfigscale{1}{figs/shadow_4n}.  \] 
\end{proof}

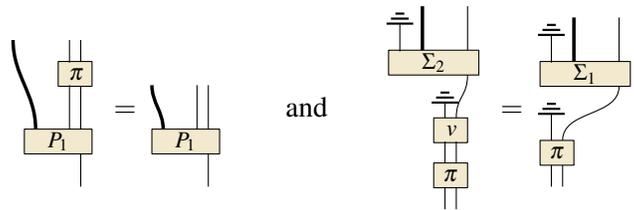
\end{document}